\documentclass[a4paper,12pt,reqno]{amsart}

\usepackage[margin=3cm]{geometry}
\usepackage{amsfonts,amsmath,amssymb,amsthm}
\usepackage{array}
\usepackage{booktabs}
\usepackage{graphicx}
\usepackage{ragged2e}
\usepackage{paralist}
\usepackage[backref=page]{hyperref}

\hypersetup{%
pdftitle={Memoryless Near-Collisions, Revisited},
pdfsubject={Computer Science},
pdfauthor={Mario Lamberger, Elmar Teufl},
pdfkeywords={hash functions, memoryless near-collisions, covering codes},
pdfborder={0 0 0},colorlinks=true}

\def\Z{\mathbb Z}
\def\Hash{H}
\def\wt#1{\qopname\relax o{w}(#1)}
\def\dist#1#2{\qopname\relax o{d}(#1,#2)}
\def\diam#1{\qopname\relax o{diam}(#1)}
\def\card#1{\lvert#1\rvert}
\def\abs#1{\lvert#1\rvert}

\def\setsep{\,|\,}

\def\ie{\textit{i.e.}}
\def\eg{\textit{e.g.}}
\def\cf{\textit{cf.}}

\newtheorem{thm}{Theorem}
\newtheorem{prop}[thm]{Proposition}
\newtheorem{lem}{Lemma}
\newtheorem{cor}{Corollary}
\theoremstyle{definition}
\newtheorem{defi}{Definition}
\newtheorem{rem}{Remark}

\begin{document}

\title{Memoryless Near-Collisions, Revisited}

\author{Mario Lamberger}
\address{%
Institute for Applied Information Processing and Communications,
Graz University of Technology,
Inffeldgasse 16a,
A--8010 Graz, Austria.}
\email{mario.lamberger@iaik.tugraz.at}

\author{Elmar Teufl}
\address{%
Mathematisches Institut,
Eberhard Karls Universit\"at T\"ubingen,
Auf der Morgenstelle 10,
D--72076 T\"ubingen, Germany.}
\email{elmar.teufl@uni-tuebingen.de}

\date{19 September 2012}

\begin{abstract}
In this paper we discuss the problem of generically finding near-collisions for
cryptographic hash functions in a memoryless way. A common approach is to
truncate several output bits of the hash function and to look for collisions of
this modified function. In two recent papers, an enhancement to this approach
was introduced which is based on classical cycle-finding techniques and
covering codes. This paper investigates two aspects of the problem of
memoryless near-collisions. Firstly, we give a full treatment of the trade-off
between the number of truncated bits and the success-probability of the
truncation based approach. Secondly, we demonstrate the limits of cycle-finding
methods for finding near-collisions by showing that, opposed to the collision
case, a memoryless variant cannot match the query-complexity of the
``memory-full'' birthday-like near-collision finding method.
\end{abstract}


\maketitle

\section{Introduction}
\label{sec:intro}

The field of hash function research has developed significantly in the light of
the attacks on some of the most frequently used hash functions like MD4, MD5
and SHA-1. As a consequence, academia and industry started to evaluate
alternative hash functions, \eg\ in the SHA-3 initiative organized by
NIST~\cite{sha3COMP}. During this ongoing evaluation, not only the three
classical security requirements \emph{collision resistance}, \emph{preimage
resistance} and \emph{second preimage resistance} are considered. Researchers
look at (semi-)free-start collisions, near-collisions, distinguishers, etc.
A `behavior different from that expected of a random oracle' for the hash
function is undesirable as are weaknesses that are demonstrated only for the
compression function and not for the full hash function.

Coding theory and hash function cryptanalysis have gone hand in hand for quite
some time now, where a crucial part of the attacks is based on the search for
low-weight code words in a linear code
(\cf\ \cite{cryptoBihamC04,cryptoChabaudJ98,imaPramstallerRR05} among others).
In this paper, we want to elaborate on a newly proposed application of coding
theory to hash function cryptanalysis. In \cite{dcc_nc,sacryptLambergerR10}, it
is demonstrated how to use covering codes to find near-collisions for hash
functions in a memoryless way. We also want to refer to the recent
paper~\cite{Gordon2010Optimal} which considers similar concepts from the
viewpoint of locality sensitive hashing.

In all of the following, we will work with binary values, where we identify
$\{0,1\}^n$ with $\Z_2^n$. Let ``$+$'' denote the $n$-bit exclusive-or
operation. The Hamming weight of a vector $v\in\Z_2^n$ is denoted by
$\wt{v} = \card{\{i \setsep v_i = 1\}}$ and the Hamming distance of two vectors
by $\dist{u}{v} = \wt{u + v}$. The Handbook of Applied
Cryptography~\cite[page~331]{bookMenezesOV96} defines \emph{near-collision
resistance} of a hash function $\Hash$ as follows:

\begin{defi}[Near-Collision Resistance]\label{d:nearcollision}
It should be hard to find any two inputs $m$, $m^*$ with $m\neq m^*$ such that
$\Hash(m)$ and $\Hash(m^*)$ differ in only a small number of bits:
\begin{equation}\label{eq:NC_def}
\dist{\Hash(m)}{\Hash(m^*)} \leq \epsilon.
\end{equation}
\end{defi}

For ease of later use we also give the following definition:

\begin{defi}\label{d:eps-near}
A message pair $m, m^*$ with $m\neq m^*$ is called an
\emph{$\epsilon$-near-collision} for $\Hash$ if \eqref{eq:NC_def} holds.
\end{defi}

Collisions can be considered a special case of near-collisions with the
parameter $\epsilon = 0$. The generic method for finding collisions for a given
hash function is based on the \emph{birthday paradox} and attributed to
Yuval~\cite{Yuval1979HowToSwindle}. There are well established cycle-finding
techniques (due to Floyd, Brent, Nivasch,
\cf\ \cite{Brent1980Improved,Knuth1997TheArtOf2,Nivasch2004CycleDetection})
that remove the memory requirements from an attack based on the birthday
paradox (see also \cite{jocOorschotW99}). These methods work by repeated
iteration of the underlying hash function where in all of these applications
the function is considered to behave like a random mapping
(\cf\ \cite{FlajoletO1989Random,Harris1960Probability}).

In \cite{dcc_nc,sacryptLambergerR10}, the question is raised whether or not the
above mentioned cycle-finding techniques are also applicable to the problem of
finding near-collisions. We now briefly summarize the ideas of
\cite{dcc_nc,sacryptLambergerR10}.

Since Definitions~\ref{d:nearcollision} and \ref{d:eps-near} include collisions
as well, the task of finding near-collisions is easier than finding collisions.
We now want to have a look at generic methods to construct near-collisions
which are more efficient than the generic methods to find collisions.

In the following, let $B_{r}(x) := \{y\in\Z_2^n \setsep \dist{x}{y} \le r\}$
denote the \emph{Hamming ball} (or \emph{Hamming sphere}) around $x$ of radius
$r$. Furthermore, we denote by
$S_n(r) := \card{B_r(x)} = \sum_{i=0}^r\binom{n}{i}$ the cardinality of any
$n$-dimensional Hamming ball of radius $r$.

A simple adaption of the classical table-based birthday attack for finding
$\epsilon$-near-collisions is to start with an empty table, randomly select a
message $m$ and compute $\Hash(m)$ and then test whether the table contains an
entry $(H(m)+\delta,m^*)$ for some $\delta \in B_\epsilon(0)$ and arbitrary
$m^*$. If so, the pair $(m, m^*)$ is an $\epsilon$-near-collision. If not,
$(\Hash(m),m)$ is added to the table and repeat. Then, we know the following:

\begin{lem}[\cite{dcc_nc}]\label{lem:memory_NC}
Let $\Hash$ be an $n$-bit hash function. If we assume that $\Hash$ acts like a
random mapping, the average number of messages that we need to hash and store
in a table-based birthday-like attack before we find an
$\epsilon$-near-collision is $O( 2^{n/2} S_n(\epsilon)^{-1/2} )$.
\end{lem}

\begin{rem}
We want to note that in this paper we are measuring the complexity of a problem
by counting (hash) function invocations. This constitutes an adequate measure
in the case of the memoryless algorithms in this paper, however the real
computational complexity of the table-based algorithm above is dominated by the
memory access, as the problem of searching for an $\epsilon$-near-collision in
the table is much harder than testing for a collision.
\end{rem}

The first straight-forward approach to apply the cycle-finding algorithms to
the problem of finding near-collisions is a truncation based approach.

\begin{lem}\label{lem:plain_trunc}
Let $\Hash$ be an $n$-bit hash function. Let
$\tau_\epsilon\colon\Z_2^n\to\Z_2^{n-\epsilon}$ be a map that truncates
$\epsilon$ bits from its input at predefined positions. If we assume that
$\tau_\epsilon\circ \Hash$ acts like a random mapping, we can apply a
cycle-finding algorithm to the map $\tau_\epsilon\circ \Hash$ to find an
$\epsilon$-near-collision in a memoryless way with an expected complexity of
about $2^{(n-\epsilon)/2}$.
\end{lem}

\begin{proof}
Under the assumptions of the lemma, the results from
\cite{FlajoletO1989Random,Harris1960Probability} are applied to a random
mapping with output length $n-\epsilon$.
\end{proof}

\section{A Thorough Analysis of the Truncation Approach}
\label{sec:trunc}

As indicated in \cite{dcc_nc}, a simple idea to improve the truncation based
approach is to truncate more than $\epsilon$ bits. That is, in order to find an
$\epsilon$-near-collision we simply truncate $\mu$ bits with $\mu > \epsilon$.
A cycle-finding method applied to $\tau_\mu\circ\Hash$ has an expected
complexity of $2^{(n-\mu)/2}$ and deterministically finds two messages $m,m^*$
such that $\dist{\Hash(m)}{\Hash(m^*)} \le \mu$. However, we can look at the
probability that these two messages $m,m^*$ satisfy
$\dist{\Hash(m)}{\Hash(m^*)} \le \epsilon$ which is
$2^{-\mu}\sum_{i=0}^{\epsilon} \binom{\mu}{i} = 2^{-\mu} S_\mu(\epsilon)$.

For a truly memoryless approach, multiple runs of the cycle-finding algorithm
are interpreted as independent events. Therefore, the expected complexity to
find an $\epsilon$-near-collision can be obtained as the product of the
expected complexity to find a cycle, and the expected number of repetitions of
the cycle-finding algorithm, \ie\ the reciprocal value of the probability that
a single run finds an $\epsilon$-near-collision. In other words, we end up with
an expected complexity of
\begin{equation}\label{eq:opt_proj}
2^{(n+\mu)/2} S_\mu(\epsilon)^{-1} = 2^{(n+\mu)/2} \,
	\left(
		\sum_{i=0}^{\epsilon} \binom{\mu}{i}
	\right)^{-1}
\end{equation}

\begin{rem}\label{rem:trunc_dcc}
In \cite{dcc_nc}, the above approach was already proposed with
$\mu=2\epsilon+1$. In this case \eqref{eq:opt_proj} results in a complexity of
\[ 2^{(n+2\epsilon+1)/2} S_{2\epsilon+1}(\epsilon)^{-1}
	= 2^{(n+1)/2-\epsilon}, \]
which clearly improves upon Lemma~\ref{lem:plain_trunc}. Here we have used that
$S_{2\epsilon+1}(\epsilon)
	= \frac12 S_{2\epsilon+1}(2\epsilon+1)
	= 2^{2\epsilon}$.
\end{rem}

An interesting question that now arises is to find the number of truncated bits
$\mu$ that constitutes the best trade-off between a larger $\mu$, \ie\ a
faster cycle-finding part, and a higher number of repetitions for this
probabilistic approach. In other words, we would like to determine the value of
$\mu$ which minimizes \eqref{eq:opt_proj} for a given $\epsilon$. Analogously,
we can search for an integer $\mu > \epsilon$ such that for a given $\epsilon$
the expression $2^{-\mu/2} S_\mu(\epsilon)$ is maximized. For small values of
$\epsilon$, values for $\mu$ were already computed in \cite{dcc_nc} by an
exhaustive search. In this section, we want so solve this problem analytically.

We first show a result that tells us something about the behavior of the
sequence of real numbers
\begin{equation}\label{eq:seq}
a_\mu := 2^{-\mu/2} S_\mu(\epsilon)
	= 2^{-\mu/2} \sum_{i=0}^{\epsilon} \binom{\mu}{i}.
\end{equation}
We want to note that based on the origin of the problem, we are only interested
in values $a_\mu$ for $\mu > \epsilon$. Our analysis is still valid starting
with $\mu=1$. We will need the following two properties of sequences:

\begin{defi}\label{def:log_unimod}
Let $a_\mu$ be a real-valued sequence.
\begin{enumerate}[(i)]
\item A sequence $a_\mu$ is called \emph{unimodal} in $\mu$,
	if there exists an index $t$ such that $a_1\le a_2 \le \dots \le a_t$
	and $a_t \ge a_{t+1} \ge a_{t+2} \ge \dots$
	The index $t$ is called a \emph{mode} of the sequence.
\item A sequence $a_\mu$ is called \emph{log-concave}, if
	$a_\mu^2 \ge a_{\mu-1}a_{\mu+1}$
	holds for every $\mu$. If $\ge$ is replaced by $>$,
	we speak of a \emph{strictly log-concave} sequence.
\end{enumerate}
\end{defi}

\begin{lem}\label{prop:unimod}
The sequence $a_\mu$ defined in \eqref{eq:seq} is strictly log-concave and
therefore also unimodal.
\end{lem}

\begin{proof}
It is a well known fact that a log-concave sequence is also unimodal, \cf\ for
example \cite{Stanley1986Unimodal}. So in order to show that \eqref{eq:seq} is
strictly log-concave we have to show that for any $\epsilon\ge 1$,
\begin{equation}\label{eq:start_ineq}
\sum_{i=0}^\epsilon \sum_{j=0}^\epsilon \binom\mu i \binom\mu j
	> \sum_{i=0}^\epsilon \sum_{j=0}^\epsilon
		\binom{\mu-1}{i}\binom{\mu+1}{j}
\end{equation}
holds. By using the recursion for the binomial coefficient twice, we can
transform the inequality \eqref{eq:start_ineq} into
\begin{align*}
\sum_{i=0}^\epsilon\sum_{j=0}^\epsilon
	\Biggl[\binom{\mu-1}{i}+\binom{\mu-1}{i-1}\Biggr] \binom\mu j
> \sum_{i=0}^\epsilon\sum_{j=0}^\epsilon
	\binom{\mu-1}{i} \Biggl[\binom\mu j+\binom{\mu}{j-1}\Biggr],
\end{align*}
which boils down to the inequality
\[ \binom\mu\epsilon \sum_{i=0}^{\epsilon-1} \binom{\mu-1}{i}
	> \binom{\mu-1}{\epsilon}\sum_{i=0}^{\epsilon-1} \binom{\mu}{i}. \]
By direct computation using the definition of the binomial coefficient, it is
easy to see that each summand on the left is strictly larger than the
respective summand on the right, simply because $\epsilon > i$.
\end{proof}

The strict log-concavity guarantees us the existence of at most two adjacent
indices for which the sequence $a_\mu$ attains its global maximum. But if there
would be an index $t$, such that $a_t=a_{t+1}$ is maximal, the definition of
the sequence $a_\mu$ in \eqref{eq:seq} shows that this would imply the
existence of two positive integers $a,b$ such that $a = \sqrt2\,b$, which is
clearly not possible. Therefore, the mode of the sequence is indeed unique.

In order to find the mode of $a_\mu$, we have to investigate some properties of
truncated sums of binomial coefficients. There are well known bounds for the
sum $S_\mu(\epsilon)$, which yield upper and lower bounds for the optimal value
of $\mu$. As we are interested in an asymptotically correct approximation for
the optimal $\mu$, we need to derive an asymptotic expansion of
$S_\mu(\epsilon)$ which seems to be hard to find in the literature.
Notationally, we use $f(\mu) \sim g(\mu)$ if
$\lim_{\mu\to\infty} f(\mu)/g(\mu) = 1$ and $f(\mu) \asymp g(\mu)$ if there
exist positive $c_1,c_2,\mu_0$ such that
$c_1\cdot\abs{g(\mu)} \le \abs{f(\mu)} \le c_2\cdot\abs{g(\mu)}$ for all
$\mu\ge\mu_0$.

\begin{prop}\label{prop:asym_sum}
Let $S_\mu(\epsilon) = \sum_{k=0}^\epsilon \binom \mu k$ and define
$\alpha := \frac\epsilon\mu$. If we assume, that there exist constants
$c_1,c_2$ such that $0 < c_1 \le \alpha \le c_2 < \frac 12$, then we have
\begin{equation}\label{eq:twoterm}
S_\mu(\epsilon) = \binom\mu\epsilon \cdot
	\biggl(
		\frac{\mu-\epsilon}{\mu-2\epsilon}
		- \frac{2\epsilon(\mu-\epsilon)}{(\mu-2\epsilon)^3}
		+ O(\mu^{-2})
	\biggr),
\end{equation}
for $\epsilon,\mu \to \infty$ and thus
\[ S_\mu(\epsilon) \sim
	\frac{\mu-\epsilon}{\mu-2\epsilon} \cdot \binom\mu\epsilon. \]
\end{prop}

\begin{proof}
For $k\le \epsilon$ we have
\begin{equation}\label{eq:binq}
\binom\mu k
	= \binom\mu\epsilon \prod_{i=0}^{\epsilon-k-1}
		\frac{\epsilon-i}{\mu-k-i}
	\le \binom\mu\epsilon \cdot
		\biggl(\frac{\epsilon}{\mu-\epsilon}\biggr)^{\epsilon-k}.
\end{equation}
Because of the requirements in the proposition we have
\[ \frac{\epsilon}{\mu-\epsilon} = \frac{\alpha}{1-\alpha}
	\le \frac{c_2}{1-c_2} < 1. \]
For sake of notation we set $\beta:=\frac{\alpha}{1-\alpha}$ and
$c:=\frac{c_2}{1-c_2}$. This then leads to
\begin{equation}\label{eq:sm}
\begin{aligned}
\binom\mu\epsilon
	&\le S_\mu(\epsilon)
	\le \binom\mu\epsilon \sum_{k=0}^\epsilon
		\biggl(\frac{\epsilon}{\mu-\epsilon}\biggr)^{\epsilon-k} \\
	&\le \binom\mu\epsilon \sum_{j=0}^\infty
		\biggl(\frac{\epsilon}{\mu-\epsilon}\biggr)^j
	= \frac{\mu-\epsilon}{\mu-2\epsilon} \cdot \binom\mu\epsilon
	\le \frac{1}{1-c} \cdot \binom\mu\epsilon.
\end{aligned}
\end{equation}
From equation \eqref{eq:sm} we learn that
$S_\mu(\epsilon) \asymp \binom\mu\epsilon$.

The following can be seen as a discrete version of Laplace's method to
approximate integrals (\cf\ \cite{deBruijnAsymptotic1981}).
\begin{align*}
S_\mu(\epsilon) = \sum_{k=0}^\epsilon \binom \mu k
	= \sum_{0\le k\le \epsilon-r} \binom \mu k
		+ \sum_{\epsilon-r<k\le\epsilon} \binom \mu k
	= S_\mu(\epsilon-r) + \sum_{0\le k<r} \binom{\mu}{\epsilon-k},
\end{align*}
where $r = r(\mu)$ is such that $r = o(\mu)$ for $\mu\to\infty$. We will
determine $r$ later.

Because of \eqref{eq:binq} and \eqref{eq:sm} we obtain
\[ S_\mu(\epsilon-r) \asymp \binom{\mu}{\epsilon-r}
	= \binom\mu\epsilon \cdot O(c^r). \]
This implies
\[ S_\mu(\epsilon) = \binom\mu\epsilon \cdot
	\Biggl( \sum_{0\le k<r}
		\prod_{i=0}^{k-1} \frac{\epsilon-i}{\mu-\epsilon+k-i}
		+ O(c^r) 
	\Biggr). \]
We now have a closer look at the product above:
\[ \prod_{i=0}^{k-1} \frac{\epsilon-i}{\mu-\epsilon+k-i} =
	\exp\Biggl(
		\sum_{i=0}^{k-1}
		\log\frac{\alpha-\frac i \mu}{1-\alpha+\frac k \mu-\frac i \mu}
	\Biggr). \]
For $x,y$ close to $0$ we have
\[ \log\frac{\alpha+x}{1-\alpha+y}
	= \log\beta + \frac1\alpha\cdot x
		- \frac1{1-\alpha}\cdot y + O(x^2 + y^2). \]
Since $0\le i<k<r$ and $r=o(\mu)$ we conclude
\begin{align*}
\log\frac{\alpha-\frac i \mu}{1-\alpha+\frac k \mu-\frac i \mu}
	= \log\beta - \frac{1}{(1-\alpha)}\cdot\frac k \mu
		- \frac{(1-2\alpha)}{\alpha(1-\alpha)}\cdot\frac i \mu
		+ O\biggl(\frac{k^2}{\mu^2}\biggr),
\end{align*}
where the error term is uniform in $0\le k< r$. With this we get
\begin{align*}
\prod_{i=0}^{k-1} \frac{\epsilon-i}{\mu-\epsilon+k-i}
&= \beta^k \exp\biggl(
		\frac{1-2\alpha}{2\alpha(1-\alpha)} \cdot \frac k \mu
		- \frac1{2\alpha(1-\alpha)} \cdot \frac{k^2}{\mu}
		+ O\biggl(\frac{k^3}{\mu^2}\biggr)
	\biggr) \\
&= \beta^k \, \biggl(
		1 + \frac{1-2\alpha}{2\alpha(1-\alpha)} \cdot \frac k \mu
		- \frac1{2\alpha(1-\alpha)} \cdot \frac{k^2}{\mu}
		+ O\biggl(\frac{k^3}{\mu^2}\biggr)
	\biggr).
\end{align*}
In total we obtain that $S_\mu(\epsilon)\big/\binom\mu\epsilon$ is equal to
\[ \sum_{0\le k<r} \beta^k \,
	\biggl(
		1 + \frac{1-2\alpha}{2\alpha(1-\alpha)} \cdot \frac k \mu
		- \frac1{2\alpha(1-\alpha)} \cdot \frac{k^2}{\mu}
		+ O\biggl(\frac{k^3}{\mu^2}\biggr)
	\biggr) \]
up to an error term which is bounded by $O(c^r)$. Since
\[ \sum_{0\le k<r} \beta^k \cdot \frac{k^3}{\mu^2} = O(\mu^{-2}) \]
and
\[ \sum_{k\ge r} \beta^k \,
	\biggl(
		1 + \frac{1-2\alpha}{2\alpha(1-\alpha)} \cdot \frac k \mu
		- \frac1{2\alpha(1-\alpha)} \cdot \frac{k^2}{\mu}
	\biggr)
	= O(r^2 c^r), \]
it follows that $S_\mu(\epsilon)\big/\binom\mu\epsilon$ is equal to
\[ \sum_{k\ge0} \beta^k \,
	\biggl(
		1 + \frac{1-2\alpha}{2\alpha(1-\alpha)} \cdot \frac k \mu
		- \frac1{2\alpha(1-\alpha)} \cdot \frac{k^2}{\mu}
	\biggr)
	+ O(\mu^{-2} + r^2 c^r). \]
Simplifying the infinite sum above yields
\[ S_\mu(\epsilon) = \binom\mu\epsilon \cdot
	\biggl(
		\frac{1-\alpha}{1-2\alpha}
		- \frac{2\alpha(1-\alpha)}{(1-2\alpha)^3} \cdot \frac 1 \mu
		+ O(\mu^{-2} + r^2 c^r)
	\biggr). \]
We now choose $r=r(\mu)=(\log \mu)^2$, since then $r^2 c^r = o(\mu^{-2})$,
which readily implies the statement using the definition of $\alpha$.
\end{proof}

The results of the Lem.~\ref{prop:unimod} and Prop.~\ref{prop:asym_sum} can now
be combined in the following way. We are interested in the behavior of $a_\mu$,
that is,
\[ a_\mu = 2^{-\mu/2} S_\mu(\epsilon)
	= 2^{-\mu/2} \sum_{i=0}^\epsilon \binom \mu i. \]
We have already seen that there will be a unique mode $t$ for the sequence.
Until this index, we have $a_{\mu+1}/a_{\mu} \ge 1$ and for all following
values of $\mu$, we have $a_{\mu+1}/a_{\mu} \le 1$. If we evaluate the
fraction, we get
$a_{\mu+1}/a_{\mu} = S_{\mu+1}(\epsilon)/(\sqrt2\,S_{\mu}(\epsilon))$.
From the recurrence relation of the binomial coefficient we get the analogous
recurrence relation for $S_\mu(\epsilon)$, namely
$S_{\mu+1}(\epsilon) = S_{\mu}(\epsilon) + S_{\mu}(\epsilon-1)
	= 2S_{\mu}(\epsilon)-\binom\mu\epsilon$.
If we use this in the above equation we end up with
\begin{equation}\label{eq:eq_for_proof}
\frac{a_{\mu+1}}{a_{\mu}} = \sqrt2\,
	\Biggl(
		1-\frac{\binom\mu\epsilon}{2S_\mu(\epsilon)}
	\Biggr)
\end{equation}
If we now use the asymptotic expansion in \eqref{eq:twoterm} we can compute an
approximation for $\mu = \mu(\epsilon)$ such that an optimum for
\eqref{eq:opt_proj} is found.

\begin{thm}\label{th:trunc_opt}
Let $\Hash$ be a hash function producing an $n$-bit hash value and let
$\epsilon\ge1$ be given. Let $\tau_\mu\colon\Z_2^n \to \Z_2^{n-\mu}$ be a map
that truncates $\mu$ fixed bits from an $n$-bit value, and suppose we apply a
cycle-finding algorithm to $\tau_\mu\circ\Hash$, which is assumed to act like a
random mapping. Then, there exists a unique optimal choice
$\mu = \mu(\epsilon) > \epsilon$ to find an $\epsilon$-near-collision. For
large $\epsilon$, we have
\begin{equation}\label{eq:opt_mu}
\mu(\epsilon) = (2+\sqrt2\,)(\epsilon-1) + O(\epsilon^{-1}).
\end{equation}
\end{thm}

\begin{proof}
Substituting the lower bound
\[ S_\mu(\epsilon)
	\ge \binom\mu\epsilon + \binom\mu{\epsilon-1}
	= \frac{\mu+1}{\mu+1-\epsilon} \binom\mu\epsilon \]
and the upper bound of \eqref{eq:sm} in \eqref{eq:eq_for_proof} implies that
the mode $t$ of the sequence $a_\mu$ is bounded by
$(1+\sqrt2)\epsilon - 1 \le t \le (2+\sqrt2) \epsilon$.
For values of $\mu$ in the domain above we may use Prop.~\ref{prop:asym_sum},
since the quotient $\epsilon/\mu$ is easily seen to be bounded in the right
way. Furthermore, $\mu\asymp\epsilon$ and $\mu-2\epsilon\asymp\epsilon$. For
large values of $\epsilon$ we infer from
\[ S_\mu(\epsilon) = \binom\mu\epsilon \cdot
	\biggl(
		\frac{\mu-\epsilon}{\mu-2\epsilon} + O(\epsilon^{-1})
	\biggr), \]
that the mode $t$ must satisfy the equation
\[ 1 = (2-\sqrt2)
	\biggl(
		\frac{t-\epsilon}{t-2\epsilon} + O(\epsilon^{-1})
	\biggr). \]
Solving this equation yields $t = (2+\sqrt2)\epsilon + O(1)$. Now let us try to
obtain further terms of the asymptotic expansion of $t$ using bootstrapping
(see for instance \cite{deBruijnAsymptotic1981}). Using the full strength of
Prop.~\ref{prop:asym_sum} implies that the equation
\[ 1 = (2-\sqrt2)
	\biggl(
		\frac{t-\epsilon}{t-2\epsilon}
		- \frac{2\epsilon(t-\epsilon)}{(t-2\epsilon)^3}
		+ O(\epsilon^{-2})
	\biggr) \]
must be satisfied by the mode $t$. Using the ansatz
$t = (2+\sqrt2)\epsilon + r$, where $r = O(1)$, yields
\[ 2\bigl(1+\sqrt2\bigr) \Bigl( (3-2\sqrt2)r
	+ (2-\sqrt2) \Bigr) \epsilon^2
	+ O(\epsilon) = 0. \]
Hence we get $r=-(2+\sqrt2) + O(\epsilon^{-1})$ and
$t=(2+\sqrt2\,)(\epsilon-1) + O(\epsilon^{-1})$ which corresponds to
$\mu(\epsilon)$.
\end{proof}

We want to note that in both, Prop.~\ref{prop:asym_sum} and
Th.~\ref{th:trunc_opt}, it is possible to compute an arbitrary number of terms
of the asymptotic expansions \eqref{eq:twoterm} and \eqref{eq:opt_mu}. We end
this section with Table~\ref{t:mu} demonstrating the quality of the
approximation of \eqref{eq:opt_mu}. The actual values for $\mu(\epsilon)$ are
produced by an exhaustive search and for simplicity, \eqref{eq:opt_mu} is
replaced with $\lceil(2+\sqrt2)(\epsilon-1)\rceil$.

\begin{table}[ht]
\begin{center}
\caption{Comparison of $\mu(\epsilon)$ and
	$\mu^*(\epsilon) := \lceil(2+\sqrt2)(\epsilon-1)\rceil$
	for certain values of $\epsilon$.}
\label{t:mu}
\medskip
\footnotesize
\begin{tabular}{@{}c*{12}{@{\hskip1mm}>{\hfill}p{6mm}}@{}}
\toprule
$\epsilon$
	& 1 & 2 & 3 &  4 & \dots & 8  & 9  & 10 & \dots &  98 &  99 & 100 \\
\midrule
$\mu(\epsilon)$
	& 2 & 5 & 8 & 11 & \dots & 25 & 28 & 32 & \dots & 332 & 335 & 339 \\
$\mu^*(\epsilon)$
	& 0 & 4 & 7 & 11 & \dots & 24 & 28 & 31 & \dots & 332 & 335 & 339 \\
\bottomrule
\end{tabular}
\end{center}
\end{table}

\section{Limitations of Memoryless Near-Collisions}
\label{sec:limits}

A drawback to the truncation based solution is of course that we can only find
$\epsilon$-near-collisions of a limited shape (depending on the fixed bit
positions), so only a fraction of all possible $\epsilon$-near-collisions can
be detected, namely $S_{\mu}(\epsilon)/S_n(\epsilon)$.

To improve upon this, \cite{dcc_nc,sacryptLambergerR10} had the idea is to
replace the projection $\tau_\epsilon$ by a more complicated function $g$,
where $g$ is the decoding operation of a certain \emph{covering code}
$\mathcal{C}$. Let $R=R(C)$ be the \emph{covering radius} of a code
$\mathcal{C}$, that is
$R(\mathcal{C}) = \max_{x \in \Z_2^{n}} \min_{c \in \mathcal{C}} \dist{x}{c}$.

\begin{thm}[\cite{dcc_nc}]\label{th:coding}
Let $\Hash$ be a hash function of output size $n$. Let $\mathcal{C}$ be a
covering code of the same length $n$, size $K$ and covering radius
$R(\mathcal{C})$ and assume there exists an efficiently computable map $g$ such
that $g\colon\Z_2^n \to \mathcal{C}$, where $x \mapsto c$ with
$\dist{x}{c} \leq R(\mathcal{C})$. If we further assume that $g\circ\Hash$ acts
like a random mapping, in the sense that the expected cycle and tail lengths
are the same as for the iteration of a truly random mapping on a space of size
$K$, then we can find $2R(\mathcal{C})$-near-collisions for $\Hash$ with a
complexity of about $\sqrt{K}$ and with virtually no memory requirements.
\end{thm}

An extensive amount of work in the theory of covering codes is devoted to
derive upper and lower bounds for $K$ (when $n$ and $R$ are given) and to
construct codes achieving these bounds
(\cf\ \cite{bookCoveringCodes1997,Struik1994AnImprovementOf,Wee1988Improved}).
The authors of \cite{sacryptLambergerR10} have investigated a class of
efficient codes suitable for the approach outlined in Th.~\ref{th:coding}. The
approach via covering codes constitutes an improvement over the purely
truncation based approach. However, (depending on $\epsilon$) the
query-complexity of the approach outlined in Th.~\ref{th:coding} is larger than
the expected query-complexity of the table-based birthday method,
\cf\ Lem.~\ref{lem:memory_NC}.

\begin{rem}\label{rem:code_prob}
We briefly want to mention the possibility of considering a probabilistic
version of the covering code approach in an analogous manner to the approach in
Sec.~\ref{sec:trunc}. In other words, what is the probability to find a
$(2R-1)$-near-collision if the covering radius is $R$? This problem has also
been studied in \cite{dcc_nc} with the outcome that in general, finding a
closed expression like \eqref{eq:opt_proj} is beyond reach. Numerical
experiments for relevant values of $n$ and $\epsilon = 2R$ show, that
increasing the covering radius is rarely bringing an improvement. We use
\cite[Eq.~(20)]{dcc_nc} together with the optimal solution from
\cite{sacryptLambergerR10} to compute complexities for small values of
$\epsilon$ in Table~\ref{t:rho_general}.
\end{rem}

The limitations of the covering code approach are inherent to the \emph{sphere
covering bound}, which states that $K \ge 2^{n}/S_n(R)$
(\cf\ \cite{bookCoveringCodes1997}). Since we use codes with covering radius
$R$ to find $2R$-near-collisions, that is, $\epsilon = 2R$, the sphere covering
bound implies that the size $K$ of the code has to be larger than
$K \ge 2^n/S_n(R) \gg 2^n / S_n(2R)$, where the latter would be the desired
quantity to match the complexity of Lem.~\ref{lem:memory_NC} to find an
$\epsilon$-near-collision.

In the following, we want to investigate, if there are other possibilities to
choose a mapping $g$ such that collisions for $g\circ\Hash$ imply
$\epsilon$-near-collisions for $\Hash$. In \cite{dcc_nc} it was shown, that the
``perfect'' mapping $g$ is beyond reach:

\begin{lem}[\cite{dcc_nc}]\label{lem:no_g}
Let $\epsilon \ge 1$, let $\Hash$ be a hash function and let $g$ be a function
such that
\[ \dist{\Hash(m)}{\Hash(m^*)}
	\leq \epsilon \Leftrightarrow g(\Hash(m))
	= g(\Hash(m^*)) \]
holds. Then, $g$ is a constant map and
$\dist{\Hash(m)}{\Hash(m^*)} \leq \epsilon$ for all $m, m^*$.
\end{lem}

So the best we can hope for is a mapping $g\colon\Z_2^n \to \Z_2^k$ that
satisfies
\begin{equation}\label{eq:eps_inj}
g(y) = g(y') \Rightarrow \dist{y}{y'} \leq \epsilon,
\end{equation}
for all $y,y' \in \Z_2^n$. If we recall the requirements of
Th.~\ref{th:coding}, it was stated that $g\circ\Hash$ should act like a random
mapping in order to have the expected cycle and tail lengths of the iteration
of $g\circ\Hash$ to be the same as for a truly random mapping on a space of
size $2^k$.

We formalize this in the following lemma. For this, we assume that the hash
function $\Hash$ acts like a random mapping from a large domain
$D \simeq \Z_2^\ell$ to $\Z_2^n$ (since most hash standards define a maximum
input length). First, we need yet another definition:

\begin{defi}\label{def:balanced}
Let $D,I$ be finite domains. We call a function $g\colon D \to I$
\emph{balanced}, if $\card{I}$ divides $\card{D}$ and for all $z\in I$ we have
$\card{g^{-1}(z)} = \card{D}/\card{I}$.
\end{defi}

\begin{lem}\label{lem:balanced}
Let $\Hash\colon \Z_2^\ell \to \Z_2^n$ be a random mapping. Furthermore,
consider a function $g\colon \Z_2^n \to \Z_2^k$ with $k \le n$. Then, $g$ is
balanced if and only if $g\circ\Hash\colon \Z_2^\ell \to \Z_2^k$ is a random
mapping.
\end{lem}

\begin{proof}
Let $g$ be balanced, that is, for all $z\in\Z_2^k$ we have
$\card{g^{-1}(z)} = 2^{n-k}$. The sets $P_z := g^{-1}(z)$ for all $z\in\Z_2^k$
define a disjoint partition of $\Z_2^n$ of size $\card{P_z}=2^{n-k}$ and $g$ is
constant on each set $P_z$.

Now let $\Hash$ be drawn uniformly at random from the set of all functions
$\Z_2^\ell \to \Z_2^n$, that is, for any function $h\colon\Z_2^\ell \to \Z_2^n$
we have $\mathbb{P}(\Hash = h) = 2^{-n 2^\ell}$. For a given
$h'\colon \Z_2^\ell \to \Z_2^k$, we now want to compute the probability
$\mathbb{P}(g\circ\Hash = h')$, for which we get
\begin{equation}\label{eq:probs}
\begin{aligned}
\mathbb{P}(g\circ\Hash = h')
	&= 2^{-n 2^\ell} \card{\{h\colon\Z_2^\ell \to \Z_2^n \setsep
		g(h(x)) = h'(x) \text{ for all } x\in\Z_2^\ell \}} \\
	&= 2^{-n 2^\ell} \card{\{h\colon\Z_2^\ell \to \Z_2^n \setsep
		h(x) \in P_{h'(x)} \text{ for all } x\in\Z_2^\ell \}} \\
	&= 2^{-n 2^\ell} 2^{(n-k) 2^\ell}
	= 2^{-k 2^\ell},
\end{aligned}
\end{equation}
because $\card{P_z}=2^{n-k}$ for all $z$. In other words, $g\circ\Hash$ is a
random mapping.

Now assume that $g\circ\Hash$ is a random mapping. That means, that for every
$h'\colon\Z_2^\ell \to \Z_2^k$ we have
$\mathbb{P}(g\circ\Hash = h') = 2^{-k 2^\ell}$. This stays true, if we choose
$h'$ to be one of the $2^k$ constant functions. If we argue along the same
lines as in \eqref{eq:probs}, we get
\[ 2^{-k 2^\ell}
	= 2^{-n 2^\ell} \card{\{h\colon \Z_2^\ell \to \Z_2^n \setsep
		g(h(x)) = c \text{ for all } x\in\Z_2^\ell \}} \]
for all $c \in \Z_2^k$. Again, with $P_c = g^{-1}(c)$, we have
\[ 2^{(n-k) 2^\ell}
	= \card{\{h\colon \Z_2^\ell \to \Z_2^n \setsep
		h(x) \in P_c \text{ for all } x\in\Z_2^\ell \}}. \]
This leaves us with $\card{P_c} = \card{g^{-1}(c)} = 2^{n-k}$ for all
$c\in\Z_2^k$, and thus, $g$ is balanced.
\end{proof}

Lem.~\ref{lem:balanced} teaches us, that in a memoryless near-collision
algorithm based on the iteration of the concatenation of the hash function
$\Hash$ and a function $g$, additionally to the requirement \eqref{eq:eps_inj}
we need $g$ also to be balanced. In the remaining part of this section, we want
to show that this limits our choices basically to the known candidates for $g$.

For the proof of the next proposition, we will need a lemma which goes back to
a conjecture by Erd\H{o}s. The solution of this problem by Kleitman in
\cite{Kleitman1966OnACombinatorial}, was further investigated in
\cite{Bezrukov1987OnTheDescription}. Let $\diam{A}$ be the diameter of a set
$A \subset \Z_2^n$, \ie, $\diam{A} := \max_{x,y\in A} \dist{x}{y}$. We now
collect the results of Th.~1 and Th.~2 of \cite{Bezrukov1987OnTheDescription}
in the following lemma:

\begin{lem}\label{lem:diam}
Let $s$ be a non-negative integer.
\begin{enumerate}[\normalfont (i)]
\item The Hamming balls $B_s(x)$ for any $x \in \Z_2^n$ are the sets of maximal
	size among all sets $A\subset \Z_2^n$ with $\diam{A} = 2s < n-1$.
\item The sets $B_s(x) \cup B_s(y )$ for any $x,y \in \Z_2^n$ with
	$\dist{x}{y}=1$ are sets of maximal size among all sets
	$A\subset \Z_2^n$ with $\diam{A} = 2s+1 < n-1$.
\end{enumerate}
\end{lem}

With this auxiliary result, we can now formulate the main result of this
section.

\begin{thm}\label{th:impossible}
Let $1 \le \epsilon < \frac n2$ be given and let $g\colon \Z_2^n \to I$ be a
balanced function satisfying \eqref{eq:eps_inj}, that is,
\[ g(y) = g(y') \Rightarrow \dist{y}{y'} \le \epsilon \]
for all $y,y'\in\Z_2^n$. Then, $\card{I}$ must satisfy
\begin{equation}\label{eq:bound_k}
\card{I} \ge 2^n / S_n(\lceil \epsilon/2\rceil)).
\end{equation}
\end{thm}

\begin{proof}
In the proof of Lem.~\ref{lem:balanced} we have seen, that the balancedness of
$g$ implies a disjoint partition $\bigcup_z P_z$ of $\Z_2^n$ where the size of
each set $P_z$ is $2^n/\card{I}$. The sets $P_z$ are exactly such that
$g(x) = z$ for all $x\in P_z$. Taking property \eqref{eq:eps_inj} into account,
we need that $\diam{P_z} \le \epsilon$. Therefore, Lem.~\ref{lem:diam} teaches
us that if $\epsilon$ is even, we have $2^n/\card{I} \le S_n(\epsilon/2)$ and
$2^n/\card{I} \le S_n(\frac{\epsilon-1}{2}) + \binom{n-1}{(\epsilon-1)/2}$,
for odd $\epsilon$. Since
$S_n(\frac{\epsilon-1}{2}) + \binom{n-1}{(\epsilon-1)/2}
	\le S_n(\frac{\epsilon+1}{2})$,
we can unify the expressions to \eqref{eq:bound_k}.
\end{proof}

As a consequence we get the following corollary:

\begin{cor}\label{cor:limits}
Let $\Hash$ be an $n$-bit hash function, let $1 \le \epsilon < \frac n 2$ and
let $g$ be a balanced function satisfying \eqref{eq:eps_inj}. Then, the
complexity to find an $\epsilon$-near-collision by applying a cycle-finding
algorithm to the concatenation $g\circ\Hash$ is bounded from below by
$\Omega(2^{n/2} S_n(\lceil \epsilon/2 \rceil)^{-1/2})$.
\end{cor}

\begin{table}[ht]
\begin{center}
\caption{Methods for finding $\epsilon$-near-collisions of an $n$-bit hash
	function $H$.}
\label{t:methods}
\medskip
\footnotesize
\begin{tabular}{@{}>{\RaggedRight}p{27mm}>{\RaggedRight}p{28mm}%
	>{\RaggedRight}p{36mm}>{\RaggedRight}p{48mm}@{}}
\toprule
\textbf{short explanation}
	& \textbf{memory}
	& \textbf{complexity}
	& \textbf{remarks} \\
\midrule
cycle finding approach applied to an $\epsilon$-truncation of $H$
	& negligible \newline (memory is only required for cycle finding)
	& $2^{(n-\epsilon)/2}$
	& \cf\ Lemma~\ref{lem:plain_trunc} and \cite{Harris1960Probability}; \\
\midrule
cycle finding approach applied to an $2\epsilon+1$-truncation of $H$
	& negligible \newline (memory is only required for cycle finding)
	& $2^{(n+1)/2-\epsilon}$
	& \cf\ Remark~\ref{rem:trunc_dcc} and \cite{dcc_nc};
		(A) in Table~\ref{t:rho_general}; \\
\midrule
cycle finding approach applied to an optimized $\mu$-truncation of $H$
		($\mu>\epsilon$)
	& negligible \newline (memory is only required for cycle finding)
	& $2^{(n+\mu)/2} S_\mu(\epsilon)^{-1}$
	& optimal $\mu=\mu(\epsilon)$ is unique and
		$\mu\sim(2+\sqrt2)(\epsilon-1)$,
		\cf\ Theorem~\ref{th:trunc_opt};
		(B) in Table~\ref{t:rho_general}; \\
\midrule
table based approach
	& a table of exponential size in $n$ for the pairs $(m,H(m))$
	& $2^{n/2} S_n(\epsilon)^{-1/2}$
	& \cf\ Lemma~\ref{lem:memory_NC} and \cite{dcc_nc};
		(C) in Table~\ref{t:rho_general}; \\
\midrule
coding based approach
	& negligible \newline (memory is only required for coding and cycle
		finding)
	& for even $\epsilon=2R$: \newline
		$2^{(n-\ell R-r)/2}$, where \newline
		$\ell := \lfloor \log_2(n/R+1) \rfloor$, \newline
		$r := \lfloor (n-R(2^\ell-1))/2^\ell \rfloor$
	& \cf\ \cite{dcc_nc,sacryptLambergerR10};
		(D) in Table~\ref{t:rho_general}; \newline
		for odd $\epsilon$ the coding based approach for $\epsilon+1$
		is repeated until an $\epsilon$-near-collision is found,
		\cf\ Remark~\ref{rem:code_prob}; \\
\bottomrule
\end{tabular}
\end{center}
\end{table}

\begin{table}[ht]
\begin{center}
\caption{For given $\epsilon \in \{1,\dots,8\}$ and hash length
	$n\in\{160,256,512\}$, the table compares the base-2 logarithms of the
	complexities (A) -- (D) of Table~\ref{t:methods}, together with (E)
	which is the bound of Corollary~\ref{cor:limits}.}
\label{t:rho_general}
\medskip
\footnotesize
\scalebox{0.86}{\begin{tabular}{@{}p{1.41mm}*{15}{>{\hfill}p{8.01mm}}@{}}
\toprule
& \multicolumn{5}{c}{$n=160$}
& \multicolumn{5}{c}{$n=256$}
& \multicolumn{5}{c}{$n=512$} \\
\cmidrule(lr){2-6} \cmidrule(lr){7-11} \cmidrule(l){12-16}
$\epsilon$
	& (A) & (B) & (C) & (D) & (E)
	& (A) & (B) & (C) & (D) & (E)
	& (A) & (B) & (C) & (D) & (E) \\
\midrule
1
	&  79.5 &  79.4 &  76.3 &  81.9 &  76.3
	& 127.5 & 127.4 & 124.0 & 130.4 & 124.0
	& 255.5 & 255.4 & 251.5 & 258.9 & 251.5 \\
2
	&  78.5 &  78.5 &  73.2 &  76.5 &  76.3
	& 126.5 & 126.5 & 120.5 & 124.0 & 124.0
	& 254.5 & 254.5 & 247.5 & 251.5 & 251.5 \\
3
	&  77.5 &  77.5 &  70.3 &  77.5 &  73.2
	& 125.5 & 125.5 & 117.3 & 125.4 & 120.5
	& 253.5 & 253.5 & 243.8 & 253.4 & 247.5 \\
4
	&  76.5 &  76.4 &  67.7 &  74.0 &  73.2
	& 124.5 & 124.4 & 114.3 & 121.0 & 120.5
	& 252.5 & 252.4 & 240.3 & 248.0 & 247.5 \\
5
	&  75.5 &  75.2 &  65.2 &  74.0 &  70.3
	& 123.5 & 123.2 & 111.5 & 121.7 & 117.3
	& 251.5 & 251.2 & 237.0 & 249.1 & 243.8 \\
6
	&  74.5 &  74.1 &  62.8 &  71.5 &  70.3
	& 122.5 & 122.1 & 108.8 & 118.5 & 117.3
	& 250.5 & 250.1 & 233.8 & 245.0 & 243.8 \\
7
	&  73.5 &  72.9 &  60.6 &  71.3 &  67.7
	& 121.5 & 120.9 & 106.2 & 118.5 & 114.3
	& 249.5 & 248.9 & 230.7 & 245.5 & 240.3 \\
8
	&  72.5 &  71.7 &  58.5 &  69.5 &  67.7
	& 120.5 & 119.7 & 103.7 & 116.0 & 114.3
	& 248.5 & 247.7 & 227.7 & 242.0 & 240.3 \\
\bottomrule
\end{tabular}}
\end{center}
\end{table}

\section{Conclusion}
\label{sec:conclusion}

At the moment, a lot of effort is dedicated to the cryptanalysis of concrete
hash function designs. From a theoretical perspective it is still very
important to investigate generic aspects of non-random properties of hash
functions. In this paper, we have analyzed several aspects of the question of
finding near-collisions in a memoryless way. This problem has recently been
investigated in \cite{dcc_nc,sacryptLambergerR10}. All these methods rely on
the application of a cycle-finding technique to an alteration (that is,
concatenation with a new mapping) of the hash function. We have investigated
in full detail the complexity of a probabilistic version of the simple
truncation based approach. Furthermore, we have shown that the approach in
general is limited in its capabilities, in the sense, that if $g$ is such
that finding a collision for $g\circ\Hash$ implies a near-collision for
$\Hash$, the query-complexity of this approach is always higher than the
query-complexity of a birthday-like method using a table of exponential size.
A comparison of the known methods is compiled in Tables~\ref{t:methods} and
\ref{t:rho_general}. It has to be noted that in practice the real complexity of
the table-based method will be dominated by the table queries and not by the
hash computations.

\section*{Acknowledgements}

The authors wish to thank the anonymous referee for valuable comments. The work
in this paper has been supported in part by the Austrian Science Fund (FWF),
project P21936-N23 and by the European Commission under contract
ICT-2007-216646 (ECRYPT II).

\def\doi#1{\href{http://dx.doi.org/#1}{\protect\nolinkurl{doi:#1}}}
\bibliographystyle{amsplainurl}
\bibliography{nearcolls30-IPL}

\parindent=0pt

\end{document}